\documentclass{amsart}
\usepackage{amsmath,amssymb,xcolor,tikz}
\usetikzlibrary{positioning,fit,calc,cd}
\tikzstyle{block} = [draw=black, thick, text width=2cm, minimum height=1cm, align=center]  
\tikzstyle{arrow} = [thick,->,>=stealth]
\def\Z{\mathbb{Z}}

\def\R{\mathbb{R}}

\def\F{\mathcal{F}}
\newcommand{\C}{\mathcal{C}}

\newcommand{\N}{\mathbb{N}}

\newcommand{\A}{\mathcal{A}}

\newcounter{thm}
\newcounter{ex}
\newcounter{re}
\newtheorem{Theorem}[thm]{Theorem}
\newtheorem{Lemma}[thm]{Lemma}

\newtheorem{Proposition}[thm]{Proposition}

\newtheorem{remark}[thm]{Remark}
\newtheorem{Definition}[thm]{Definition}

\title[KP hierarchy, pseudo-differential operators and Yang-Mills]{Kadomtsev-Petviashvili hierarchies with 
non-formal pseudo-differential operators, non-formal solutions, and a Yang-Mills--like formulation}
\author[J.-P. Magnot and E.G. Reyes]{Jean-Pierre Magnot$^1$ and Enrique G. Reyes$^2$}

\address{\small $^1$: {LAREMA, Universit\'e d’Angers, 2 Bd Lavoisier, 
49045 Angers cedex 1, France,  and  Lyc\'ee Jeanne d'Arc, 40 avenue de Grande Bretagne, 63000 Clermont-Ferrand, 
France}, and Lepage Research Institute.
17 novembra 1, 081 16 Presov,
Slovakia}
\email{\small magnot@math.univ-angers.fr; jean-pierr.magnot@ac-clermont.fr}
\address{\small $^2$:
	Departamento de Matem\'{a}tica y Ciencia de la Computaci\'{o}n,
	Universidad de Santiago de Chile, Casilla 307 Correo 2, Santiago,
	Chile. }\email{\small enrique.reyes@usach.cl;
	e\_g\_reyes@yahoo.ca}
\begin{document}

\begin{abstract}
We start from the classical Kadomtsev-Petviashvili hierarchy posed on formal pseudo-differential 
operators, and we produce two hierarchies of non-linear equations posed on non-formal pseudo-differential 
operators lying in the Kontsevich and Vishik's odd class, one of them with values in formal pseudo-differential 
operators. We prove that the corresponding Zakharov-Shabat equations hold in this context, and we express one of 
our hierarchies as the minimization of a class of Yang-Mills action functionals on a space of 
pseudo-differential connections whose curvature takes values in the Dixmier ideal. We finish by comparing our  
Kadomtsev-Petviashvili hierarchies in terms of the solutions that they produce to the KP-II 
equation: existence, uniqueness and formality.  
\end{abstract}
\maketitle


\smallskip

\smallskip

\section*{Introduction}
For more than thirty years, solutions to the Kadomtsev-Petviashvili (KP) hierarchy have played a central role in 
geometric approaches to Mathematical Physics, not only in the theory of integrable systems (for example due to 
its links with equations such as the Burgers, Korteweg-de Vries and Boussinesq equations, see {\em e.g.,} 
\cite{Kodama}), but also for its connections with algebraic curves (see for instance \cite{M2}) and with the 
geometry of the Pressley-Segal infinite-dimensional grassmannian. Considering this last issue, the connection of 
the KP hierarchy --formulated with formal pseudodifferntial operators-- with groups involving (non-formal) 
Hilbert-Schmidt operators on a separable  Hilbert space, see \cite{Mick,PS}, remains a exciting and amazing 
procedure even after all these years. 

In our previous paper \cite{MR2018}, we solved the Cauchy problem and proved well-posedness for 
a \textit{deformed} version of a KP hierarchy constructed (for the first time, to the best our knowledge) 
with non-formal odd class pseudodifferential operators, and we exhibited its relationship with vertex-like 
formulations of a class of non-formal Fourier integral operators which lie in the Pressley-Segal unitary group 
$U_{res}.$ In the present work, we address the problem of the formulation of a standard KP hierarchy in terms of 
non-formal operators \emph{and} we aim at producing solutions that are non-formal, that is, they are not formal 
series but {\em true} smooth functions in the independent variables. We say that solutions (and henceforth, 
equations) are 
\emph{delinearized}. This leads to the definition of a non-linearized KP hierarchy with values in 
pseudo-differential operators of non-formal type. If they exist, the operator-valued solutions are proven to 
fulfill non-formal versions of the Zakharov-Shabat equations. Moreover, since operators are non-formal, this 
enables us to produce a Yang-Mills type equation which is equivalent to the Zakharov-Shabat equations. The 
classical resolution for the KP hierarchy in terms of r-matrices (the method of solutions that we used in
\cite{MR2016,MR2018}, for instance) does not hold in our context. Therefore, we can reduce there equations in two 
different ways: 
\begin{enumerate}
	\item by taking a linearized form of each object
	\item by taking the formal part of operators.  
\end{enumerate}
In the linearization process, we recover our capacity to solve the Cauchy problem by means of r-matrix methods,
 and we prove well-posedness for the \textit{non-deformed} version 
of the KP hierarchy, expressed with non-formal odd class pseudodifferential operators. 
In taking the formal part of operators, the delinearized version of the Kadontsev hierarchy does not seem to 
produce any advance. But it is an intermediate step to consider delinearized solutions to the KP-II equations.
Gathering (1) and (2), we recover the classical KP hierarchy. 

This paper is organized as follows: Section \ref{s:1} is dedicated to necessary preliminaries on non-formal 
pseudo-differential operators, focusing on the odd class of operators first described by Kontsevich and Vishik in 
\cite{KV1}. Then, in Section \ref{s:2} the announced delinearized Kadomtsev-Petviashvili hierarchy is studied,
and  the Zakharov-Shabat equations and their Yang-Mills formulation is given. In sections \ref{s:3} and \ref{s:4}, 
we prove a Birkhoff-Mulase factorization and well-posedness in the context of a KP hierarchy formulated on formal 
series in time, but for non-formal operators. In the last section, we discuss the relationship of the delinearized 
KP hierarchy with the solutions of the standard KP-II equation. 
 
\section{On non-formal pseudodifferential operators} \label{s:1}
We introduce the groups and algebras
of (non-formal!) pseudo-differential operators that we need in order to set up a meaningful KP hierarchy.  
Their choice is a delicate matter, since we aim at proving a smooth factorization result for an appropriate
group to be introduced in Section 2, what we call a {\em non-formal Birkhoff-Mulase decomposition} after \cite{PS} 
and \cite{M1,M3}. In this section $E$ is $E = S^1 \times V$ in which $V$ is a finite-dimensional complex vector 
space. The following definition  appears in \cite{MR2018}; it is taken from \cite[Section 2.1]{BGV}.

\begin{Definition} 
	The graded algebra of differential operators acting on the space of smooth sections $C^\infty(S^1,V)$ is the 
	algebra $DO(S^1,V)$ generated by:
	
	$\bullet$ Elements of $End(E),$ the group of smooth maps $E \rightarrow E$ leaving each fibre globally 
	invariant, 
	and restricting to linear maps on each fibre. This group acts on sections of $E$ via (matrix) multiplication;
	
	$\bullet$ The differentiation operators
	$$\nabla_X : g \in C^\infty(S^1,E) \mapsto \nabla_X g$$ where $\nabla$ 
	is a connection on $E$ and $X$ is a vector field on $S^1$.
\end{Definition}

Multiplication operators are operators of order $0$; differentiation operators and vector fields are operators 
of order 1. 

The algebra $DO(S^1,V)$ is graded by order: 
we denote by $DO^k(S^1,V)$, $k \geq 0$, the space of differential operators of order less or equal than $k$.
$DO(S^1,V)$ is a subalgebra of the algebra of classical pseudo-differential operators
$Cl(S^1,V),$ an algebra which is filtered by order and that contains, for example, the square root of the 
Laplacian, its inverse, and all trace-class operators on $L^2(S^1,V).$ Basics facts on pseudo-differential 
operators can be found in \cite{Gil}. 
Following \cite{Ma2016}, we assume henceforth that $S^1$ is equipped with charts such that the changes 
of coordinates are translations. We list some of the spaces of pseudo-differential operators that we will use in 
this work.


\begin{itemize}
\item $ PDO (S^1,V) $  
(resp.  $ PDO^o (S^1,V)$, resp. $Cl(S^1,V)$) is the space of pseudo-differential operators 
(resp. pseudo-differential operators of order $o$; resp. classical pseudo-differential operators) acting on 
$C^\infty(S^1,V).$
\item $Cl^o(S^1,V)= PDO^o(S^1,V) \cap Cl(S^1,V)$ is the space of classical pseudo-differential 
operators 
of order $o$. 
\item $Cl^{0,\ast}(S^1,V)$ is the group of units of $Cl^0(S^1,V)$. 
\end{itemize}


A topology on spaces of classical pseudo-differential operators has been described by Kontsevich and Vishik in 
\cite{KV1}; see also \cite{CDMP,PayBook,Scott} for other descriptions.
We use all along this work the Kontsevich-Vishik topology: for fixed $o\in \Z,$ the Kontsevich-Vishik topology  
defines a Fr\'echet topology on $Cl^o(S^1,V).$
We also recall that the set of smoothing operators (see e.g.
\cite{Gil,Scott}) is the two-sided ideal of $PDO(S^1,E)$ defined by 
$$ PDO^{-\infty}(S^1,V) = \bigcap_{o \in \Z} PDO^o(S^1,V) \; .$$
Therefore, the following quotients are well defined:
$$\mathcal{F}PDO(S^1,V) = PDO(S^1,V) / PDO^{-\infty}(S^1,V)\; ,$$ 
$$\F Cl(S^1,E) = Cl(S^1,E) /  PDO^{-\infty}(S^1,E)\; ,$$
$$ \quad \F Cl^o(S^1,V) = Cl^o(S^1,V)  / PDO^{-\infty}(S^1,V)\; .$$
\noindent {The script font $\F$ stands for} {\it  formal pseudo-differential operators}. 
The quotient algebra $\mathcal{F}PDO(S^1,E)$ can be identified with the set of formal symbols, see \cite{BK}, 
and the identification is a morphism of $\mathbb{C}$-algebras for the usual multiplication on formal symbols 
(see e.g. \cite{Gil}). Finally, we denote by ${\mathcal F}Cl^{0,*}(S^1,V)$ the group of units of the algebra 
${\mathcal F}Cl^{0}(S^1,V)$, and
by $Diff_+(S^1)$ the group of orientation preserving diffeomorphisms of $S^1$. 

\begin{Theorem} 
The groups $Cl^{0,*}(S^1,V)$, $Diff_+(S^1)$, and $\mathcal{F}Cl^{0,*}(S^1,V)$ are regular Fr\'echet Lie groups.
\end{Theorem}

This theorem is already well-known: it is noticed in \cite{Ma2006,MR2018}, by applying the results of 
\cite{Gl2002}, that the group $Cl^{0,*}(S^1,V)$ (resp.$\F Cl^{0,*}(S^1,V)$) is open in $Cl^0(S^1,V)$ 
(resp. $\F Cl^{0}(S^1,V)$) and also that it is a regular {\em Fr\'echet} Lie group. Also, it follows from 
\cite{Ee,Om} that $Diff_+{S^1}$ is open in the Fr\'echet manifold $C^\infty(S^1,S^1)$. This fact makes it a 
Fr\'echet manifold and, following \cite{Om}, a regular Fr\'echet Lie group.

\begin{Definition} \label{d7}
	A classical pseudo-differential operator $A$ on $S^1$ is called odd class
	if and only if for all $n \in \Z$ and all $(x,\xi) \in T^*S^1$ we have:
	$$ \sigma_n(A) (x,-\xi) = (-1)^n  \sigma_n(A) (x,\xi)\; ,$$
	in which $\sigma_n$ is the symbol of $A$.
\end{Definition}

This particular class of pseudo-differential operators was introduced by Kontsevich and Vishik in 
\cite{KV1,KV2}. Odd class operators are also called ``even-even class'' operators, see \cite{Scott}. We choose 
to follow the terminology of the first two references. Hereafter, the notation $Cl_{odd}$ will refer to odd 
class classical pseudo-differential operators.  The following proposition summarizes Lemma 3.4 and 
Proposition 3.5 of \cite{MR2018}.

\begin{Proposition} 
	The algebra $Cl_{odd}^0(S^1,V)$ is a closed subalgebra of $Cl^0(S^1,V)$. 
	The space of units $Cl_{odd}^{0,*}(S^1,V)$ is 
	\begin{itemize}
		\item An open subset of $Cl^0(S^1,V)$ and,
		\item A regular Fr\'echet Lie group.
	\end{itemize} 
\end{Proposition}

\smallskip

By the symmetry property stated in Definition \ref{d7}, an odd class pseudo-differential operator $A$ has a 
partial symbol of non-negative order $n$ of the form
\begin{equation} \label{alfa}
\sigma_{n}(A)(x,\xi) = \gamma_n(x) (i\xi)^n \, ,
\end{equation} 
where $\gamma_n \in C^\infty(S^1,L(V))$. {\em This consequence of Definition $\ref{d7}$ allows 
us to check the following direct sum decomposition:}

\begin{Proposition}\label{SD}
	$$Cl_{odd}(S^1,V) = Cl_{odd}^{-1}(S^1,V) \oplus DO(S^1,V)\; .$$  
\end{Proposition} 

The second summand of this expression will be specialized in the sequel to differential operators having symbols 
of order 1. Because of (\ref{alfa}), we can understand these symbols as elements of $Vect(S^1)\otimes Id_V.$

\section{A class of delinearized Kadomtsev-Petviashvili hierarchies with values in non-formal PDOs.} \label{s:2}

Let us assume that $t_1, t_2, \cdots, t_n , \cdots,$ are an infinite number of different formal variables. 
{{} In the usual approach to the KP hierarchy, see for instance \cite{M3}, we use formal series in 
$t_1, \cdots , t_n, \cdots$ as coefficients of formal pseudo-differential operators. Now, these formal series }
can be understood as a specific set of smooth functions 
on the algebraic sum 
$$T= \bigoplus_{n \in \mathbb{N}^*}(\mathbb{R}\,t_n)$$ for the product topology on $T$, which identifies it as a 
topological vector space with the classical inductive limit $\underset{\longrightarrow}{\lim} \R^n = \R^\infty $,  
see \cite{Ma2013,MR2016,GMW2023}. The space $T$ carries natural differentiations with respect to the variables 
$t_1, t_2, \cdots$, so that we can define the space $C^\infty(T, Cl_{odd}(S^1,V)).$ {{} This space will 
be the starting point of our constructions.} The 
decomposition appearing in Proposition \ref{SD}, namely,
$Cl_{odd}(S^1,V) = Cl_{odd}^{-1}(S^1,V) \oplus DO(S^1,V)$, extends naturally to $C^\infty(T, Cl_{odd}(S^1,V))$;
we have
\begin{equation} \label{desc1}
C^\infty(T, Cl_{odd}(S^1,V)) = C^\infty(T, Cl_{odd}(S^1,V))_S \oplus C^\infty(T, Cl_{odd}(S^1,V))_D\; ,
\end{equation}
with 
$$ C^\infty(T, Cl_{odd}(S^1,V))_S = C^\infty(T, Cl_{odd}^{-1}(S^1,V))$$ and 
$$ C^\infty(T, Cl_{odd}(S^1,V))_D = C^\infty(T, DO(S^1,V))\; .$$

{ For later use, we also observe that in the same way, we can consider the two-sided ideal }
$PDO^{-\infty}(S^1,V)$ of the algebra $Cl_{odd}(S^1,V),$ and we 
can form the two-sided ideal $$C^\infty(T, PDO^{-\infty}(S^1,V))$$ of the algebra $C^\infty(T, Cl_{odd}(S^1,V)).$

\smallskip

{ The decomposition (\ref{desc1}) allows us to set up} the {\em delinearized Kadomtsev-Petviashvili 
hierarchy} { on $C^\infty(T, Cl_{odd}(S^1,V))$} with initial value $L_0 \in  Cl_{odd}(S^1,V)\,$:

\begin{equation}
	\label{KP-delin-nonformal}
	\left\{ \begin{array}{ccl}  L(0) & = & L_0 \\
	\frac{dL}{dt_n} & = & \left[ L^n_D , L \right] = - \left[ L^n_S , L \right] \; , \quad \quad n=1,2, \cdots .
	\end{array} \right.
\end{equation}


Now we remark that in this context we also obtain zero 
curvature equations, in analogy with the classical case considered in \cite{M1}. 

\begin{Theorem} \label{ZS0}
Let us assume that $L \in C^\infty(T, Cl_{odd}(S^1,V))$  solves the delinearized KP hierarchy 
$(\ref{KP-delin-nonformal})$. We consider 
	$Z = \sum_k L^{\; k} dt_k, $ and we decompose this one-form as 
	$Z = Z_D - Z_S$. Then, the following two equations hold:
	$$ dZ_D - \left[Z_D , Z_D\right]=0 \quad \mbox{ and } \quad dZ_S - \left[Z_S , Z_S\right]=0\; .$$
\end{Theorem}
\begin{proof}
	Let $L$ be a solution to our KP hierarchy (\ref{trueKP}), so that for all $k \geq 1$ we have 
	$$\frac{d L}{dt_k} = \left[ L^k_D, L \right] = - \left[ L^k_S, L \right]\; .$$
	Then, for $n,m \geq 1$ we have,
	\begin{eqnarray*}
		\frac{d L^n}{dt_m} &=& \sum_{k = 0}^{n-1} L^k  \frac{d L}{dt_m} L^{n-k-1} \;
		= \; [L^m_D,L^n] = {{} -  } [L^m_S,L^n]\; .
	\end{eqnarray*}
	Therefore, 
	\begin{eqnarray*}
		\frac{d L^n}{dt_m} - \frac{d L^m}{dt_n} 
		& = & -[L^m_S,L^n] + [L^n_S,L^m] \\
		& = & -2 [L^m_S,L^n_S] -  [L^m_S,L^n_D] + [L^n_S,L^m_D]\; ,
	\end{eqnarray*}
	and, in the same way, 
	\begin{eqnarray*}
		\frac{d L^n}{dt_m} - \frac{d L^m}{dt_n} 
		& = & [L^m_D,L^n] - [L^n_D,L^m] \\
		& = & 2 [L^m_D,L^n_D] +  [L^m_D,L^n_S] - [L^n_D,L^m_S]
	\end{eqnarray*}
	Therefore, gathering the two expressions we just obtained, we have 
	\begin{eqnarray*}
		2 \left(\frac{d L^n}{dt_m} - \frac{d L^m}{dt_n} \right)
		& = & -2 [L^m_S,L^n_S] -  [L^m_S,L^n_D] + [L^n_S,L^m_D] +2 [L^m_D,L^n_D] \\ 
		&   & + [L^m_D,L^n_S] - [L^n_D,L^m_S]  \\
		& = & -2 [L^m_S,L^n_S] + 2 [L^m_D,L^n_D]  \; ,  
	\end{eqnarray*}
and thus, by considering the $D-$ and $S-$ projections, we obtain the following zero curvature equations in our
	non-formal context:
\begin{equation} \label{zcr1}
	\frac{d L^n_D}{dt_m} - \frac{d L^m_D}{dt_n} =  [L^m_D,L^n_D] \; ,
\end{equation}
	and
\begin{equation} \label{zcr2}
	\frac{d L^n_S}{dt_m} - \frac{d L^m_S}{dt_n} = -[L^m_S,L^n_S] \; ,
\end{equation}
	for arbitrary $n, m \geq 1$. These equations are equivalent to
	$$dZ_D - \left[Z_D  , Z_D\right] = 0 $$
	and
	$$dZ_S - \left[Z_S , Z_S\right] = 0$$
	respectively.  
\end{proof}

\begin{remark}
	Mulase \cite{M1} gives the same formulas with the notation $dZ_D - \frac{1}{2}\left[Z_D  , Z_D\right]$ and $dZ_S - \frac{1}{2}\left[Z_S , Z_S\right].$ In Mulase's approach, the coefficient $\frac{1}{2}$ is justified by skew-symmetrized operations on differential forms with values in a non-commutative algebra, while our notations integrate the coefficient $\frac{1}{2}$ in the formula. 
\end{remark}

{{ } As an application of Theorem \ref{ZS0}, we now point out that the ``zero curvature equations" 
(\ref{zcr2}) can be understood in a way that reminds 
us of a Yang-Mills action defined on an infinite-dimensional bundle over the space of independent variables $T$.}
Let $\tau_2$ be  the two-sided ideal of Hilbert-Schmidt operators. Let us consider the trivial principal bundle
 $P= T \times G$ where $G$ is a group with Lie 
algebra that contains $\tau_2$, and $\C(P)=\Omega^1(T,\tau_2)$   
the space of its connection 1-forms. We have the following identification for each $n-$cube:  
$$
[-k;k]^n \sim \{(t_1,t_2,...)\in T \, | \, \forall i \in \N_n, -k \leq t_i \leq k \hbox{ and } 
\forall i \in \N - \N_n, t_i=0\}\; ,$$
{where $\forall n \in \N^*, $  $\N_n = \{1,2,3...n\}$.}

Let $\theta = \sum_{i} \theta_{i} dt_i \in \C(P)$ and let $F(\theta) = d\theta - [\theta,\theta]$ be its 
curvature, that we decompose as 
$$F(\theta) = \sum_{i<j} F(\theta)_{i,j} dt_i \wedge dt_j\; ,$$ 
where 
$$F(\theta)_{i,j} \in C^\infty(T,\tau_2) \; .$$
Thus, we identify the curvature $F(\theta)$ with a 2-form on the Adjoint bundle.  
Then, we see that the following conditions are equivalent: 
\begin{equation} 
F(\theta) = 0 
\end{equation}

\begin{equation} 
\forall (i,j) \in (\N^*)^2, tr\left(F(\theta)_{i,j}F(\theta)_{i,j}^* \right)=0 
\end{equation}

\begin{equation} \label{YMZScurvature} 
\forall (k,n,i,j) \in (\N^*)^4 \hbox{ with } 
	i<j, \quad YM_{k,n}(\theta)_{i,j} = \int_{[-k;k]^n} tr\left(F(\theta)_{i,j}F(\theta)_{i,j}^*\right) = 0\; , 
\end{equation}
because $(a,b) \mapsto tr( a b^* )$ is positive-definite, and in particular non-degenerate, on $\tau_2.$

Since $Cl^{-1}_{odd}(S^1,V) \subset \tau_2$, the constructions that lead to (\ref{YMZScurvature}) and 
Theorem \ref{ZS0} yield the following result:

\begin{Theorem}
	The zero-curvature Zakharov-Shabat equations $(\ref{zcr2})$, for the delinearized KP hierarchy 
	$(\ref{KP-delin-nonformal})$, 
	are equivalent to 
	$$ \forall (k,n,i,j) \in (\N^*)^4 \hbox{ with } 
	i<j, \quad YM_{k,n}(Z_S)_{i,j} = \int_{[-k;k]^n} tr\left(F(Z_S)_{i,j}F(Z_S)_{i,j}^*\right) = 0.$$
\end{Theorem}

\begin{remark}
	All constructions of this section are performed without any assumption on the operator $L,$ neither in 
	terms of order of the operator nor in terms of dressing by a Sato-like operator. However, for the rest of the 
	paper, we assume that $L$ is of constant order 1 and that it can be written as 
$$
L = S_0 \frac{d}{dx} S_0^{-1}
$$ 
where 
$$ 
S_0 \in Cl^{0,*}_{odd}(S^1,V) \hbox{ and } S_0 - Id_V \in Cl^{-1}_{odd}(S^1,V)\; .
$$
	The {{ } necessity} of these conditions {{ } for non-formal pseudo-differential operators} has 
	to be investigated; we leave these questions for future works. 
\end{remark}

\section{Birkoff-Mulase decomposition for Fr\'echet Lie groups of series of non-formal PDOs} \label{s:3}

Let us now fix $L_0 \in Cl_{odd}(S^1,V).$ We consider infinite jets of 
$${{ }  L \in  C^\infty(T,Cl_{odd}(S^1,V) ) } $$ 
with $L(0) = L_0\,$. 

We could write $j^\infty L $ in order to denote the infinite jet of $L$, but since we consider infinite jets 
only in this section, we will note it $L$ instead of $j^\infty L .$ {{ } These jets are expressed as 
series in the T-variables with coefficients in $Cl_{odd}(S^1,V).$ They can be also understood as formal smooth 
functions in the $T-$variables.}
Series in $T-$variables are equipped with the standard valuation used for 
the KP-hierarchy: $val_T(t_n) = n$, see e.g. \cite{M1}.

\begin{remark} We note that our classes differ from the ones appearing in \cite{PPV2024}, where formal pseudodifferential operators, that are series in a formal derivation variable $\xi,$ with monomial coefficients in trace-class ideals, are considered.  In our context, formal series are with respect to the $T-$variables. 
	\end{remark}

\begin{Definition} \label{dfs}
The set of odd class $T-$pseudo-differential operators is the set {{ } of infinite jets} 
	\begin{equation} \label{alpha}
	Cl_{T,odd}(S^1,V) \subset Cl_{odd}(S^1,V)[[t_1,t_2,...]] \; ,
	\end{equation}
	such that the coefficient $a_n$ of a monomial of order $n$ lies in $Cl_{odd}^n(S^1,V).$
\end{Definition}

 This definition tries to capture explicitly (but formally) the ``time dependence'' of the dependent variables 
appearing in the KP hierarchy. The idea of using power series in infinitely many times is due to Mulase, see 
\cite{M3}, who used it in his analysis of the standard KP hierarchy.  In Mulase's context, $T$-series 
(with coefficients in a given ring) form the algebra of coefficients of the 
formal pseudo-differential operators appearing in the KP hierarchy. The fact that we can impose some control over 
their growth is a technical observation at the core of the proof of the deep  
factorization theorem appearing in \cite{M3} (a review of this theorem appears in \cite{ER2013}). We have used 
different versions and adaptations of Mulase's insight in \cite{Ma2013,MR2016,MR2018}. 

\medskip

We state the following result on the structure of $Cl_{T,odd}(S^1,V)$: 

\begin{Theorem} \label{fs}
	The set $Cl_{T,odd}(S^1,V)$ is a Fr\'echet algebra, and its group of units, $Cl_{T,odd}^*(S^1,V)$, 
	is a regular Fr\'echet Lie group. The Lie algebra of  $Cl_{T,odd}^*(S^1,V)$ is precisely $Cl_{T,odd}(S^1,V).$
\end{Theorem}

Before proving this theorem, we need a technical lemma which is, essentially, a special case of a 
similar result that is valid in a more general context, see e.g. \cite{GMW2023}. Since here we state it only in 
the framework of Fr\'echet manifolds, we give a complete proof with optimal arguments on the Fr\'echet setting. 

We assume that we have at hand a sequence of Fr\'echet vector spaces $(A_n)_{n \in \N^*},$ ordered by inclusion,  
and such that the canonical inclusion map is continuous.    
Moreover, we assume that there exist bilinear smooth maps $$A_n \times A_m \rightarrow A_{n+m}$$ that define a 
multiplication on 
$$ \mathcal{A} = \prod_{n \in \N^*} A_n\; ,$$ 
and we identify $\mathcal{A}$ with 
$$\left\{ \sum_{n=1}^{+ \infty} a_n q^n \, | \, a_n \in A_n \right\}\; ,$$
in which $q$ is some formal variable. 
We remark that multiplication and addition on $\mathcal{A}$ fit with classical addition and 
multiplication on formal series, and 
that we can even understand $\mathcal{A}$ as an ideal of a wider space of formal series 
$$\left\{ a_0 q^0 + a \, | \, (a_0,a) \in \C \times \mathcal{A} \right\}$$
with natural  $q-$valuation.
We obtain the following result: 

\begin{Lemma} \label{regulardeformation2}
The group $1 + \mathcal{A}$ is a regular Fr\'echet Lie group with Lie algebra $\mathcal{A}.$ 
\end{Lemma}
\noindent
\begin{proof} The set $1+\mathcal{A}$ is clearly a Fr\'echet manifold (in fact, an affine space) since 
$\mathcal{A}$ is a Fr\'echet vector space. It is also a Fr\'echet Lie group because of the classical 
formulas of multiplication and inversion on formal series, and its Lie algebra can be clearly 
identified with $\mathcal{A}$. 

Let us now show that the exponential of paths $C^\infty([0;1],\A) \rightarrow C^\infty([0;1], 1+\A)$ is a smooth
map. 
Let $v \in C^\infty([0,1],\A)$. We take $s \in [0;1]$ and $j = \lfloor ns \rfloor.$ We define
$$u_n(s) = \left(1+\left(s - \frac{j}{n}\right)v\left(\frac{j}{n}\right)\right) 
\prod_{i = 1}^{j}  \left(1+\frac{1}{n} v\left(\frac{j-i}{n}\right)\right).$$

\noindent We have that 
$$\lim_{n \rightarrow +\infty}\partial_s u_n(s) . u_n^{-1}(s) = \lim_{n \rightarrow +\infty} 
v\left(\frac{j}{n}\right)\left(1+\left(s - \frac{j}{n}\right)v\left(\frac{j}{n}\right)\right)^{-1} = v(s). $$
Moreover, the $A_m$ component of the product converges to a sum of integrals of the type  
$$
\int_{1\geq s_1 \geq ... \geq s_k\geq 0} \left[ \prod_{i = 1}^k v(s_i)\right]_m (ds)^k
$$
for $k \leq m$. This shows convergence to a path $u\in C^\infty([0;1];1+\A)$ satisfying 
$$ \partial_s u(s).u^{-1}(s) = v(s)$$ that smoothly depends on the path $v \in C^\infty([0;1],\A)$ in the 
Fr\'echet sense. 
\end{proof}

\smallskip

\begin{Definition}
Let $q$ be a formal parameter. 
We define the algebra of formal series 
$$Cl_q(S^1,E) = 
\left\{ \sum_{k \in \N^*} q^k a_k \,|\, \forall k \in \N^*, a_k \in \left[Cl_{T,odd}(S^1,E)\right]_k \right\},$$
in which $\left[Cl_{T,odd}(S^1,E)\right]_k$ denotes the (Fr\'echet) vector space of monomials with 
{{ } $T$-valuation} equal to $k.$
\end{Definition}
This is obviously an algebra, graded by the order (the valuation) of the variable  $q.$ Thus, setting
$$ A_n = \left\{ q^n a_n | a_n \in\left[Cl_{T,odd}(M,E)\right]_n \right\} ,$$
we state the following consequence of Lemma \ref{regulardeformation2}, changing $\mathcal{A}$ by $Cl_q(S^1,E)$ 
in this context:

\begin{Lemma}
The group $1 + Cl_q(S^1,E)$ is a regular Fr\'echet Lie group with Lie algebra $Cl_q(S^1,E).$
\end{Lemma}

A direct consequence of  {{}\cite[section 38.6, page 413]{KM}} to the short exact sequence of 
Fr\'echet Lie groups:
$$ 0 \rightarrow 1 + Cl_q(S^1,E) \rightarrow Cl^{0,*}(S^1,E) + Cl_q(S^1,E) \rightarrow Cl^{0,*}(S^1,E) 
\rightarrow 0$$  
is the following:
\begin{Lemma} \label{regulardeformation}
$Cl^{0,*}(S^1,E) + Cl_q(S^1,E)$ is a regular Lie group with Lie algebra $Cl^{0}(S^1,E) + Cl_q(S^1,E).$
\end{Lemma}

We can now finish the proof of Theorem \ref{fs}:

\begin{proof}[Proof of Theorem \ref{fs}.]
We note that the $q-$valuation and the $T-$valuation coincide, therefore one can identify 
$Cl_{T,odd}^*(S^1,V)$ with $Cl^{0,*}(S^1,E) + Cl_q(S^1,E)$, for instance by replacing formally  $q$ 
for $1$.
 Therefore, this result is essentially an application of Lemma \ref{regulardeformation}. 
\end{proof}   
%
	\begin{remark}
		The assumption $a_n \in Cl^n_{odd}$ in Definition $\ref{dfs}$ can be relaxed 
		to the condition 
		$$a_0 \in Cl^{0,*}_{odd} \hbox{ and } \forall n \in \N^*, a_n \in Cl_{odd}\; ;$$ 
		this is sufficient for having a regular Lie group. However,  a smooth Birkhoff-Mulase 
		decomposition (see Theorem \ref{non-formal-birkhoff} below) seems to fail in this more general context.
		Following \cite{MR2016}, we find that the growth conditions imposed in {\rm (\ref{alpha})}  will 
		ensure both regularity and also existence of a Birkhoff-Mulase decomposition. 
	\end{remark}

The decomposition $L = L_S + L_D$, $L_S \in Cl^{-1}_{odd}(S^1,V)$, 
$L_D \in DO^1(S^1,V)$, which is valid on $Cl_{odd}(S^1,V)$, see Proposition \ref{SD}, extends straightforwardly 
to the algebra $Cl_{T,odd}(S^1,V)$.  As a consequence, we have the following decomposition on the 
group $Cl_{T,odd}^*(S^1,V):$

\begin{Theorem} \label{non-formal-birkhoff}
    Let $S_0 \in Cl^{0,*}_{odd}(S^1,V)$ such that $S_0 - Id_V \in Cl^{-1}(S^1,V).$ Let 
    $L_0 = S_0 \left(\frac{d}{dx}\right) S^{-1}_0.$ Let 
    $U(t_1,...,t_n,...) = \exp\left(\sum_{n \in N^*} t_n (L_0)^n\right) \in Cl_{T,odd}^*(S^1,V).$ Then, 
    there exists a unique pair $(S,Y)$ such that
		\begin{enumerate}
			\item $U = S^{-1}Y,$ 
			\item $Y \in Cl_{T,odd}^*(S^1,V)_D$
			\item $S \in Cl_{T,odd}^*(S^1,V)$ and $S - 1 \in Cl_{T,odd}(S^1,V)_S.$ 
		\end{enumerate}
		Moreover, the map 
\begin{equation} \label{map}
S_0\in Cl^{0,*}_{odd}(S^1,V)\times T \mapsto \left(U(t_1,t_2,...),Y(t_1,t_2,...), S(t_1,t_2,...)\right)
		\in (Cl_{odd}^*(S^1,V))^3
\end{equation}
		is smooth.
\end{Theorem}

\noindent
{\bf Notation.} We note by {{ } $Cl^{-1,*}_{odd}(S^1,V)$} the group of {{ } ``dressing operators"} 
$S \in Cl^{0,*}_{odd}(S^1,V)$ such that 
$Id - S \in Cl^{-1}_{odd}(S^1,V).$ 

\begin{proof}
The proof consists of passing from decompositions and equations for formal pseudo-differential operators 
following the proofs of e.g. \cite{ER2013,ERMR,MR2016}, that are modelled after \cite{M1,M3}, to analogous 
properties 	for series of non-formal, odd class pseudo-differential operators. For this, we denote 
by $\sigma(A)$ the formal operator (or equivalently the asymptotic expansion of the symbol) corresponding to the 
operator $A.$ Because of the previous results in \cite{M1,M3}, see also \cite{ER2013,ERMR,MR2016}, 
there exist non-formal odd class operators $Y$ and $W$ defined up to smoothing operators such that 
	$$\sigma(U) = \sigma(S)^{-1}\sigma(Y)\; .$$
	Now, $\sigma(Y)$ is a formal series in $t_1, \cdots t_n,\cdots$ of symbols of differential operators, which   
	are in one-to-one correspondence with a series of (non-formal)
	differential operators. Thus, the operator $Y$ is uniquely defined, not up to a smoothing operators; it 
	depends 
	smoothly on $U$, and so does $S = Y U^{-1}.$ Finally, smoothness of the map (\ref{map}) follows
	from uniqueness and the corresponding result for symbols proven in \cite{MR2016}.
\end{proof}

\section{Extending the classical Kadomtsev-Petviashvili hierarchy to non-formal PDOs} \label{s:4}

We now introduce \emph{linearized form} of the Kadomtsev-Petviashvili hierarchy with non-formal pseudo-
differential operators, that can be understood from two viewpoints: 
\begin{itemize}
	\item first, by linearization of our \emph{delinearized} KP hierarchy (\ref{KP-delin-nonformal})
	\item second, by adapting our work carried out in \cite{MR2016}, see also the review \cite{ERMRR}, and 
	modifying the approach on 
	non-formal operators initiated in \cite{MR2018}.
\end{itemize}
 We make the following definition:

\begin{Definition}
	Let $S_0 \in Cl^{-1,*}_{odd}(S^1,V)$   be a non-formal dressing operator, and let 
	$L_0 = S_0 (\frac{d}{dx})S_0^{-1}.$ 
	We say that an operator 
	$$ L(t_1,t_2,\cdots) \in  Cl_{T,odd}(S^1,V)$$ 
	satisfies the
	(non formal) KP hierarchy if and only if 
	\begin{equation} \label{jph}
	\left\{\begin{array}{cl} 
	L(0) = & L_0 \\ 
	\displaystyle  \frac{d}{dt_n}L =& \left[(L^n)_D, L\right] \; , \quad \quad n \geq 1\; .
	\end{array}
	\right.
	\end{equation}
\end{Definition}

\smallskip

\noindent We can solve the initial value problem for (\ref{jph}):
 
\smallskip

\begin{Theorem} \label{hKP}
	Let $U(t_1,...,t_n,...) = \exp\left(\sum_{n \in N^*} t_n (L_0)^n\right) \in Cl_{T,odd}(S^1,V)$ (in fact, 
	$U$ belongs to $Cl_{T,odd}^*(S^1,V)$). Let $(S,Y)$ be the Birkhoff-Mulase decomposition of $U$ 
as in Theorem \ref{non-formal-birkhoff}.  
	Then, the operator $L \in Cl_{T,odd}(S^1,V)$ given by $L = S L_0 S^{-1} = Y L_0 Y^{-1}$ 
		is the unique solution of the hierarchy of equations 
		\begin{equation}
		\label{trueKP} \left\{\begin{array}{ccl}
	\displaystyle	\frac{d}{dt_n}L &=& \left[(L^n)_D(t_1, ...), L(t_1,...)\right] =  -\left[(L^n)_S(t_1,...), L(t_1, ...)\right]
	\; , \quad \quad n \geq 1 \; , \\
		L(0) & = & L_0 \; , \\
		\end{array}\right. .
		\end{equation}
		Moreover, 
 \begin{itemize}
		\item  Taking the formal symbols of the operator $L \in Cl_{T,odd}(S^1,V)$ given by 
		$L = S L_0 S^{-1} = Y L_0 Y^{-1}$, we recover the unique solution of the classical KP hierarchy on formal 
		operators of \cite{M1},
  \item By applying the scaling $t_n \mapsto h^nt_n$ and $\frac{d}{dx} \mapsto h \frac{d}{dx},$ we recover the 
  solution to the non-formal $h-$deformed KP hierarchy of \cite{MR2018}.
 \end{itemize}
\end{Theorem}	
\begin{proof}
We consider the $h-$deformed KP equation introduced in \cite{MR2018} via scaling. As proven therein, it has a 
unique solution which depends smoothly on $S_0.$ Now, this result is stated in the above reference for 
diffeological smoothness. By the inverse scaling $t_n \mapsto \frac{t_n}{h^n}$ and 
$\frac{d}{dx} \mapsto \frac{1}{h} \frac{d}{dx}$, we recover uniqueness and (diffeological) smoothness of the 
solution of our KP hierarchy (\ref{trueKP}). Following \cite{MR2016}, see also \cite{GMW2023}, diffeological 
smoothness in the context of Fr\'echet manifolds is equivalent to G\^ateaux smoothness; this observation ends the 
proof. 
\end{proof}

\smallskip

We finish this section with a remark on the zero curvature formulation of the KP hierarchy. This formulation is 
of course well-known in a formal setting, see for instance \cite{D,M1,M3}.


\begin{Theorem} \label{ZS1}
Let us assume that $L \in Cl_{T,odd}(S^1,V)$  solves the KP hierarchy $(\ref{trueKP})$. We consider 
$Z = \sum_k L^{\; k} dt_k, $ and we decompose this one-form as 
$Z = Z_D - Z_S$. Then, the following two equations hold:
$$ dZ_D - \left[Z_D , Z_D\right]=0 \quad \mbox{ and } \quad dZ_S - \left[Z_S , Z_S\right]=0\; .$$
\end{Theorem}
\begin{proof}
The proof mimicks the proof of Theorem \ref{ZS0} by replacing the pseudo-differential operator
$L \in C^\infty(T,Cl_{odd}(S^1,V))$ for $L \in Cl_{T,odd}(S^1,V).$
\end{proof}

\section{On delinearized and linearized KP-II equations} \label{s:5}

The calculations that follow are, 
formally, very similar to the ones in \cite{MRR2022}, but they differ from them (and from the ones appearing in 
any other reference, to the best of our knowledge) by the fact that the symbol
$$ 
{{ } L = \xi + \sum_{k \leq -1} u_k \xi^k  } 
$$
belongs to 
$${{ } 
C^\infty(T,\mathcal{F}Cl_{odd}(S^1,V)) = C^\infty(T,Cl_{odd}(S^1,V))/C^\infty(T,PDO^{-\infty}(S^1,V))\; , ???}
$$ 
in which the formal variable $\xi$ represents the Fourier transform of the derivation $\frac{d}{dx}$ on $S^1$.
In other words, the coefficients $u_k$ belong to $C^\infty(T,C^\infty(S^1,V))$ instead of being formal series in 
$T.$ 
 Within this setting, we can form yet another KP hierarchy, \emph{delinearized} but valued in 
 \emph{formal operators}, that reads as 
 \begin{equation}
 	\label{KP-delin-formal}
 	\left\{ \begin{array}{ccl}  L(0) & = & L_0 \\
 		\frac{dL}{dt_n} & = & \left[ L^n_D , L \right] = - \left[ L^n_S , L \right] \end{array} \right.
 \end{equation}
with $L_0 = S_0 \frac{d}{dx} S_0^{-1},$  $S_0 - Id_V \in \mathcal{F}Cl^{-1}(S^1,V)$ and 
$L \in C^\infty(T,\mathcal{F}Cl^1(S^1,V)).$ 

\smallskip

Computing as before, we deduce 
zero curvature conditions analogous to (\ref{zcr1}) in this context. Let us reduce our investigations to $V=\R,$
for illustrative purposes. The following classical calculations on formal partial symbols hold:

\smallskip

$$\begin{array}{|c|c|c|}
	\hline
	& L^2 & L^3 \\ \hline
	\sigma_3 &-&\xi^3
	\\ \hline
	\sigma_2 &\xi^2&0
	\\ \hline
	\sigma_1 &0&3 u_{-1}\xi
	\\ \hline
	\sigma_0 &2u_{-1}\xi^0& (3 u_{-2} + 3 \frac{d}{dx} u_{-1}) \xi^0 \\ \hline
\end{array}$$
Therefore,
\begin{eqnarray*}
	\left[L^2_+,L^3_+\right] & = & 	\left[\xi^2,\xi^3\right] + 3	\left[\xi^2,u_{-1}\xi\right] + 
	3 	\left[\xi^2,u_{-2}+ \frac{d}{dx} u_{-1}\right] \\
	&& + 2	\left[u_{-1},\xi^3\right]+ 6 	\left[u_{-1},u_{-1}\xi\right]  \\
	& = & (3 \frac{d^2}{dx^2} u_{-1}+ 6\frac{d}{dx} u_{-2})\xi + (-6\frac{d}{dx} u_{-1} u_{-1} + 
	\frac{d^3}{dx^3} u_{-1} + 3\frac{d^2}{dx^2}u_{-2}) \; .
\end{eqnarray*}

The ZS-equations for the pairs $(t_1,t_2),$ $(t_1,t_3)$ and $(t_2,t_3)$ read respectively as
\begin{equation}\label{t12}
	\frac{dL_+^2}{dt_1} - \frac{dL_+}{dt_2} = \left[L_+,L^2_+\right]\; ,
\end{equation}
which gives: 
\begin{equation}\label{t12-deg1}
	\frac{du_{-1}}{dt_1} = \frac{d}{dx} u_{-1}\; ;
\end{equation} 

\begin{equation}\label{t13}
	\frac{dL_+^3}{dt_1} - \frac{dL_+}{dt_3} = \left[L_+,L^3_+\right]\; ,
\end{equation}
which gives:
\begin{equation}\label{t13-deg1}
	\left\{ \begin{array}{l}\frac{du_{-1}}{dt_1} = \frac{d}{dx} u_{-1}
		\\ 
		\frac{du_{-2}}{dt_1} = \frac{d}{dx} u_{-2}  \; ; 
		\end{array} \right.
\end{equation} 

\noindent and finally for the pair $(t_2,t_3),$
\begin{equation}\label{t23}
	\frac{dL^3_+}{dt_2} - \frac{dL_+^2}{dt_3} = \left[L^2_+,L^3_+\right]\; ,
\end{equation}
which gives:
\begin{equation}\label{t23-deg1prime}
	\left\{ \begin{array}{l}\frac{du_{-1}}{dt_2} = \frac{d^2}{dx^2} u_{-1} + 2\frac{d}{dx} u_{-2}
		\\ 
		3\frac{d}{dt_2}u_{-2} + 3 \frac{d^2}{dt_2 dx} u_{-1} - 2\frac{d}{dt_3} u_{-1} = 
		-6\frac{d}{dx} u_{-1} u_{-1} + \frac{d^3}{dx^3} u_{-1} + 3\frac{d^2}{dx^2} u_{-2} \; . \end{array} \right.
\end{equation}
This system is equivalent to 
\begin{equation}\label{t23-deg1}
	\left\{ \begin{array}{l}\frac{du_{-1}}{dt_2} = \frac{d^2}{dx^2} u_{-1} + 2 \frac{d}{dx} u_{-2}
		\\ 
		3\frac{d }{dt_2}u_{-2} - 2\frac{d}{dt_3}u_{-1} = -6\frac{d}{dx} u_{-1} u_{-1} 
		-2 \frac{d^3}{dx^3} u_{-1} - 3 \frac{d^2}{dx^2} u_{-2} \; , \end{array} \right.
\end{equation}
and we stress that fact that the functions $u_{-1}$ and $u_{-2}$ are in $C^\infty(T, C^\infty(S^1,\R)).$
The results of \cite{MRR2022} imply that this instance of the KP equation {{ }can be solved by means of formal solutions $(u_{-1},u_{-2}) \in  C^\infty(S^1,\R)[[t_2,t_3]]^2$, but to the actual state of knowledge, the non-formal solutions in
$C^\infty(\R^2, C^\infty(S^1,\R))$ cannot be produced by r-matrix methods, because the Birkhoff-Mulase factorization does not apply to this case.}

\medskip

Let us summarize our work. We have obtained a ``commutative diagram of hierarchies''
that induces, via a reduction to the plane generated by $(t_2,t_3)$ variables, the same KP-II equation with its solutions either in a formal formulation,
or in non-formal formulation.

\medskip

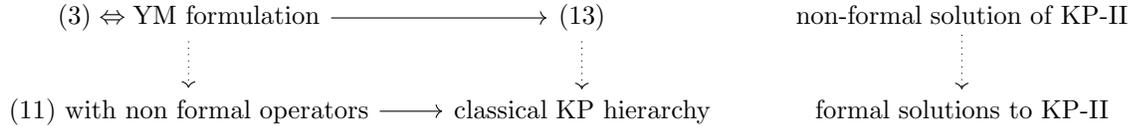
\begin{tikzcd}
	(\ref{KP-delin-nonformal}) \Leftrightarrow \hbox{YM formulation} \arrow [d,dotted] \arrow [r]&(\ref{KP-delin-formal}) \arrow [d,dotted] & \hbox{non-formal solution of KP-II} \arrow [d,dotted]\\
	(\ref{jph}) \hbox{ with non formal operators} \arrow [r] & \hbox{classical KP hierarchy} & \hbox{formal solutions to KP-II}
\end{tikzcd}

\medskip

In this diagram, the dotted arrow means taking the infinite jets, passing from non-linearized to linearized 
objects, and the plain arrows mean taking the formal parts of operators. Each line produce one type of solution, 
the first one solutions in $C^\infty(\R^2 \times S^1, \R),$ and the second one solutions in 
$C^\infty(S^1,\R)[t_2,t_3].$ 

\medskip

\paragraph{\bf Acknowledgements:} J.-P.M was partially
supported by ANID (Chile) via the FONDECYT grant \#1201894 during the first stages of this work.
He also thanks the France 2030 framework programme Centre Henri Lebesgue ANR-11-LABX-0020-01 
for creating an attractive mathematical environment.
E.G.R.'s research was partially supported by the FONDECYT grants  \#1201894 and \#1241719.


\begin{thebibliography}{99}
\bibitem{BGV} Berline, N.; Getzler, E.; Vergne, M.; \textit{Heat Kernels and Dirac Operators} Springer (2004)

	\bibitem{BK}  Bokobza-Haggiag, J.; { Op\'erateurs pseudo-diff\'erentiels sur une vari\'et\'e 
	diff\'erentiable}; {\it Ann. Inst. Fourier, Grenoble} \textbf{19,1}  125-177 (1969)
	
	
	\bibitem{CDMP} Cardona, A.; Ducourtioux, C.; Magnot, J-P.; Paycha, S.;
	Weighted traces on pseudo-differential operators and geometry on loop groups;
	\textit{Infin. Dimens. Anal. Quantum Probab. Relat. Top.} \textbf{5} no4 503-541 (2002)
	
	\bibitem{D} Dickey, L.A.; \textit{Soliton equations and Hamiltonian systems, second edition} 
	Advanced Series in Mathematical Physics $12$, World Scientific Publ. Co., Singapore (2003).
	\bibitem{Ee} Eells, J.;  A setting for global analysis
	\textit{Bull. Amer. Math. Soc.} {\bf 72} 751-807 (1966)
	
	\bibitem{ER2013} Eslami Rad, A.; Reyes, E. G.; The Kadomtsev-Petviashvili hierarchy and the Mulase 
	factorization of formal Lie groups {\it J. Geom. Mech.} {\bf 5}, no 3 (2013) 345--363.
	
	\bibitem{ERMR} Eslami Rad, A.; Magnot, J.-P.; Reyes, E. G.; The Cauchy problem of the Kadomtsev-Petviashvili 
	hierarchy with arbitrary coefficient algebra. {\it J. Nonlinear Math. Phys.} {\bf 24}:sup1, 103--120 (2017).
	
	\bibitem{ERMRR} Eslami Rad, A.; Magnot, J.-P.; Reyes, E. G.; Roubtsov, V.; Diffeologies and generalized 
	Kadomtsev-Petviashvili hierarchies. {\em Contemporary Mathematics} {\bf 794} , 211--222 (2024).
	\bibitem{GMW2023} Goldammer, N.; Magnot, J-P.; Welker, K.; On diffeologies from infinite dimensional geometry
	 to PDE constrained optimization. \emph{Contemporary Mathematics} {\bf 794}, 1-48 (2024).
 
 \bibitem{Gil} Gilkey, P;
	{\it Invariance theory, the heat equation and the Atiyah-Singer index theorem}
	Publish or Perish (1984)
	
	\bibitem{Gl2002} Gl\"ockner, H; Algebras whose groups of the units are Lie groups
	\textit{Studia Math. } \textbf{153}, no2,
	147--177 (2002).
	
	
	
	
	\bibitem{Kodama} Y. Kodama,  KP Solitons and the Grassmannians. 
	Combinatorics and Geometry of Two-Dimensional Wave Patterns. Springer Briefs in Mathematical Physics Vol. 22,
	(2017).
	
	\bibitem{KV1} Kontsevich, M.; Vishik, S.;
	{ Determinants of elliptic pseudo-differential operators} Max
	Plank Institut fur Mathematik, Bonn, Germany, preprint n. 94-30
	(1994)
	
	\bibitem{KV2}  Kontsevich, M.; Vishik, S.; Geometry of determinants of elliptic operators.
	Functional analysis on the eve of the 21st century, Vol. 1 (New Brunswick, NJ, 1993), 
	\textit{Progr. Math.} \textbf{131},173-197
	(1995)
	
	\bibitem{KM} Kriegl, A.; Michor, P.W.; \textit{The convenient setting for global analysis}
	Math. surveys and monographs \textbf{53}, American Mathematical society, Providence, USA. (2000)
	
	
	
	\bibitem{Ma2006} Magnot, J-P.; Chern forms on mapping spaces,
	\textit{Acta Appl. Math.} \textbf{91}, no. 1, 67-95 (2006).
	
	\bibitem{Ma2013} Magnot, J-P.; Ambrose-Singer theorem on diffeological bundles and complete integrability
	of KP equations. {\em Int. J. Geom. Meth. Mod. Phys.} {\bf 10}, no 9 Article ID 1350043 (2013).
	
	\bibitem{Ma2016} Magnot, J-P.; On $Diff(M)-$pseudodifferential operators and the geometry of non linear 
	grassmannians. 	\textit{Mathematics} \textbf{4}, 1; doi:10.3390/math4010001 (2016).
	
	
	\bibitem{MR2016} Magnot, J-P.; Reyes, E. G.; Well-posedness of the Kadomtsev-Petviashvili hierarchy, 
	Mulase factorization, and Fr\"olicher Lie groups. \emph{Ann. Henri Poincar\'e} {\bf 21} no. 6, 1893-1945 
	(2020). 
	
	\bibitem{MRR2022} Magnot, J-P.; Reyes, E.G.; Rubtsov, V.; On $(t_2,t_3)$-Zakharov-Shabat equations of generalized Kadomtsev-Petviashvili hierarchies. 
	\emph{J. Math. Phys.} {\bf 63} no. 9, Article ID 093501, 11 p. (2022). 
	
	\bibitem{MR2018} Magnot, J-P. ; Reyes, E. ; On the Cauchy problem for a Kadomtsev-Petviashvili hierarchy on 
	non-formal operators and its relation with a group of diffeomorphisms. \emph{Dynamics of PDEs} 21, No. 3 
	, 235--260 (2024). 
	
	\bibitem{Mick} Mickelsson, J.; \textit{Current algebras and groups}. Plenum monographs in Nonlinear Physics, 
	Springer (1989).
	
	
	
\bibitem{M1} Mulase, M.; Complete integrability of the Kadomtsev-Petvishvili equation.
	{\em Advances in Math.} \textbf{54} (1984), 57--66.
\bibitem{M2} Mulase, M.;  Cohomological structure in soliton equations and
	Jacobian varieties. {\em J. Diff. Geom.} 19, 403--430 (1984).
\bibitem{M3} Mulase, M.; Solvability of the super KP equation and a generalization of the Birkhoff
	decomposition. {\em Invent. Math.} \textbf{92}, 1--46 (1988).
	
	\bibitem{Om} Omori, H.; {\it Infinite Dimensional Lie Groups}; AMS Translations of Mathematical Monographs no 
	{\bf 158}  Amer. Math. Soc., Providence, R.I. (1997).
	
	
	\bibitem{PayBook} Paycha, S; 
	\textit{Regularised integrals, sums and traces. An analytic point of view.}
	University Lecture Series \textbf{59}, AMS (2012).
	
	\bibitem{PS} Pressley, A.; Segal, G.; \textit{Loop groups} OUP (1986).
	
	\bibitem{PPV2024} Prykarpatski, A.J.; Pukach, P.Y.; Vovk, M.I.; On the dual representation of the congruence kernels and the related Delsarte type transmutations of multidimensional differential operators. \emph{Oper. Theory Adv. Appl.} {\bf 295}, 297--315 (2024).
	
	\bibitem{Scott} Scott, S.;
	\textit{Traces and determinants of pseudodifferential operators}. 
	OUP (2010).
	
	
	
\end{thebibliography}
\end{document}